\documentclass{article}
\usepackage{graphicx,amssymb,amsmath,amsthm}
\usepackage[margin=1in]{geometry}
\usepackage{subcaption}

\newtheorem{lemma}{Lemma}[section]
\newtheorem{theorem}{Theorem}[section]

\newcommand{\dist}{\mathsf{d}}

\newcommand{\ball}{\mathsf{ball}}
\newcommand{\e}{\varepsilon}
\newcommand{\out}{\mathsf{output}}
\newcommand{\spread}{\Delta}
\newcommand{\pts}{\mathsf{points}}
\renewcommand{\dim}{\mathsf{dim}}

\title{The Product Range Search Problem in Doubling Metrics}

\author{Oliver Chubet\thanks{North Carolina State University, \texttt{oachubet@ncsu.edu}}
	\and
	Niyathi Kukkapalli\thanks{Princeton University, \texttt{nk6074@princeton.edu}}
	\and
	Anvi Kudaraya\thanks{UC Merced, \texttt{anvikudaraya417@gmail.com}}
	\and
	Donald R. Sheehy\thanks{North Carolina State University, \texttt{don.r.sheehy@gmail.com} This work was supported by the U.S. National Science Foundation under grant CCF-2349179}
}

\begin{document}
\maketitle

\begin{abstract}
    We present two data structures for a generalization of orthogonal range search called product range search in which the coordinates are replaced with arbitrary doubling metrics.
    A standard metric range search takes a query $(q,r)$ and finds all points within distance $r$ of the point $q$.
    In the product setting, we have two (or more) different metrics, say $\dist_1$ and $\dist_2$.
    A product range query $(q, r_1, r_2)$ is a point $q$ and two radii $r_1$ and $r_2$.
    The output is all points within distance $r_1$ of $q$ with respect to $\dist_1$ and all points within $r_2$ of $q$ with respect to $\dist_2$.
    In other words, it is the intersection of two searches.
    Our results apply to any number of metrics.
    
    The first data structure is a generalization of the classic range tree from computational geometry using greedy trees rather than binary trees.
    The second data structure is a single metric tree constructed on the $\ell_\infty$ product of the input metrics.
    
\end{abstract}

\section{Introduction}

The metric range search problem is as follows: given a set of points $P$ in a metric space, preprocess $P$ into a data structure where queries are of the form $(q,r)$ where $q$ is a point in $P$ and $r \ge 0$ is a real number.
Such a query returns the points in $P \cap \ball(q, r)$.
In one dimension, the metric balls are intervals and
products of intervals are axis-aligned boxes.
The corresponding range search problem is known as \emph{orthogonal range search}.
Thus, orthogonal range search is a product of one-dimensional range searches.
In this paper, we generalize the metric range search problem to product metric range searches in doubling metrics.
For two metrics, $\dist_1$ and $\dist_2$, the queries will be of the form $(q, r_1, r_2)$ and the result will be: 

\[
    \{x\in P \mid \dist_1(x,q)\le r_1 \text{ and } \dist_2(x,q) \le r_2\}.
\] 

It is the intersection of two range searches $(q, r_1)$ with respect to $\dist_1$ and $(q, r_2)$ with respect to $\dist_2$. We call this the \textbf{Product Range Search Problem}.
We consider the more general setting where there is a sequence of $m$ metrics $\dist_1,\ldots, \dist_m$ on a set $P$.

In this paper, we present and analyze two data structures for product range search in doubling metrics.
The first structure is a cascade of search trees similar to range trees used in the Euclidean setting.
The second product range tree directly uses a metric tree called a greedy tree on the product metric itself. 
However, the query algorithm must be modified to accommodate products of balls, which may not be balls in the product metric.
The two approaches give two alternatives for combining the metrics, the first addresses them in sequence, pruning the search and the second combines them for a single search.
Asymptotically, neither is strictly better than the other.

In settings where distance computations can be computed efficiently, there are several efficient data structures available, and some authors even consider range search a ``solved'' problem~\cite{agarwal1999geometric}.
However, there are still cases when range queries over the full space may be impractical due to expensive metric computations.
Recently, product metrics have been considered for cases where metric computations are expensive to compute, but the metric admits some decomposability into a product, such as is the case for the Ulam metric, discrete Fr\'{e}chet distance, and dynamic time warping~\cite{afshani2018complexity,emiris2020products}.
Other applications that can be formulated as product metrics are multi-key range queries in databases~\cite{bentley1979data}.
Search problems in product metrics has also been studied for metrics that decompose into $l_p$ products of $l_\infty$ or $l_1$ spaces~\cite{andoni2009overcoming}.
To the best of our knowledge, product range search has not been studied for products of non-normed spaces.

\subsection{Related Work}

There is a plethora of data structures for range search in Euclidean space.
For coordinate-wise comparisons, the $k$-d tree partitions the plane into rectangular regions and has query time of $O(n^{1 - 1/d} + k)$ for a $d$-dimensional query~\cite{de2008computational}.
There are also range trees for orthogonal range queries, which have a query time of $O(\log^{d-1}n + k)$ using fractional cascading~\cite{de2008computational}. 

Ball trees are the simplest hierarchical structure for range search in metric spaces.
They are defined by recursively partitioning the data set, represented as a binary tree~\cite{omohundro1989five}.
Each node of the tree represents a metric ball, storing a center and radius.
Range searches are performed by traversing the tree.
Nodes can be pruned when the entire ball is disjoint from the query ball.
Greedy trees are ball trees constructed using farthest point sampling to achieve packing and covering guarantees~\cite{chubet2023proximity}.
The algorithms we propose in this work are built with greedy trees because they are simple, implementable, and admit strong theoretical guarantees, but other similar trees could work.

\subsection{Contribution and Outline}

This work builds on general metric space range search data structures.
The background on the products of metric spaces, the doubling dimension, and standard packing bounds in these spaces is presented in Section~\ref{sec:background}.
That section also recounts the necessary background on greedy trees, which will be used throughout.

Our first data structure, the \emph{greedy range tree}, is explained and analyzed in Section~\ref{sec:cascading}. 
It uses greedy trees in place of the binary trees used in Euclidean range trees.
Given an $n$-point set $P$ and $m$ metrics $\dist_1, \dots, \dist_m$, the greedy range tree can be built in $2^{O(\delta)} n \log^{m-1} \spread$ time, where $\delta = \sum_{i=1}^m \dim(P, \dist_i)$ is the sum of the doubling dimensions (see Section~\ref{sec:doubling_dimension}), and the spread, $\spread$, is the ratio between the maximum and minimum pairwise distances. For simplicity, we assume the spread of each metric is the same, otherwise we can take the maximum.
The greedy range tree can be stored in $2^{O(\delta)} n \log^{m-1} \spread$ space.
Each query takes time
\[
    \left(\frac{1}{\e} \right)^{O(\delta)} m\log \spread+ O(k_\e),
\]
where $k_\e$ is the number of points in the output assuming the ranges are all increased by a factor of $(1+\e)$ (see Section~\ref{sec:using_approximation}).

Note that this data structure can be easily generalized to use \emph{any} ball tree.
For example, if the underlying tree is a net-tree, then the query time analysis will be similar.
However, in order to achieve the reported preprocessing time, we rely on recent results allowing us to merge greedy trees in linear time~\cite{chubet2024simple}.

In Section~\ref{sec:non-cascading} we explain and analyze product range queries in a regular greedy tree built using the product metric itself (see also Section~\ref{sec:products} for the formal definition).
The data structure can be constructed in $2^{O(\delta)}n \log n$ time and uses only $O(n)$ space.
As in the greedy range tree, the running time depends on the approximate output size $k_{\e}$ for a user defined $\e\le \frac{1}{2}$.
A product range search takes 
\[
    \left(\frac{A}{\e} \right)^{O(\delta)} \log \Delta + O(k_\e)
\]
time, where $A$ denotes the ratio between the maximum and minimum radii in the query.
Table~\ref{tab:ds_comparison} highlights the tradeoffs between the greedy range tree and the greedy tree on the product metric in terms of storage, build time, and query complexity.

\begin{table*}[ht]
    \centering
    \renewcommand{\arraystretch}{1.5} 
    \begin{tabular}{|p{3cm}|c|c|c|}
        \hline
        &\textbf{Preprocessing} & \textbf{Query Time} & \textbf{Space} \\
        \hline
        Greedy Range Tree (Section~\ref{sec:cascading}) & $2^{O(\delta)} n \log^{m-1} \spread$& $\left(\frac{1}{\e} \right)^{O(\delta)} m\log \Delta + O(k_\e)$ & $2^{O(\delta)} n \log^{m-1} \Delta$\\
        \hline
        Greedy Tree on the Product Metric (Section~\ref{sec:non-cascading}) & $2^{O(\delta)}n \log n$ & $\left(\frac{A}{\e}\right)^{O(\delta)}\log\spread + O(k_\e)$ & $O(n)$\\
        \hline
    \end{tabular}

        \caption{The preprocessing times, query times, and space used by the data structures presented in Sections~\ref{sec:cascading} and~\ref{sec:non-cascading}.
    We assume the trees are built on an $n$-point set $P$ with $m$ metrics $\dist_1, \dots, \dist_m$ with spread $\spread$.
    We define $\delta := \sum_{i=1}^m\dim(P, \dist_i)$ the sum of the doubling dimensions.
    $A$ is the aspect ratio of a query, defined as the ratio of the maximum query radius to the minimum query radius, $k_\e$ denotes the approximate output size, and $\e$ is an approximation parameter.}
    \label{tab:ds_comparison}

\end{table*}

\section{Background}
\label{sec:background}

\subsection{Product Metrics}\label{sec:products}

Let $P$ be a set of points, and $\dist_1, \dist_2, \dots, \dist_m$ metrics on $P$.
The $l_\infty$-\emph{product metric} $\dist$ is 
\[
    \dist(x, y) = \max_{1 \leq i \leq m} \dist_i(x, y).
\]
 
For clarity, when referring to a ball in metric $\dist_i$, we use the notation $\ball_i$, however when referring to a ball in the product metric itself we may omit the subscript.

\subsection{Bootstrapping Approximate Range Search}
\label{sec:using_approximation}
Given a set of points $P$ and $m$ metrics $\dist_1, \ldots, \dist_m$, 
an approximate product range search query is a tuple $(q, r_1, \dots, r_m, \e)$ containing a point $q$, radii $r_1, \ldots, r_m$, and an approximation parameter $\e$.
The $\out$ contains \emph{all} points in the intersection of the metric balls centered at a query point $q$.
Any other points in the output are \emph{approximately} contained in the intersection.
Specifically,
\[
P \cap\bigcap_{i=1}^m \ball_i(q,r_i) \subseteq \out \subset P\cap\bigcap_{i=1}^m \ball_i\left(q, (1 + \e) r_i)\right).
\]

\begin{figure}[ht]
    \center
    \includegraphics[scale=0.6]{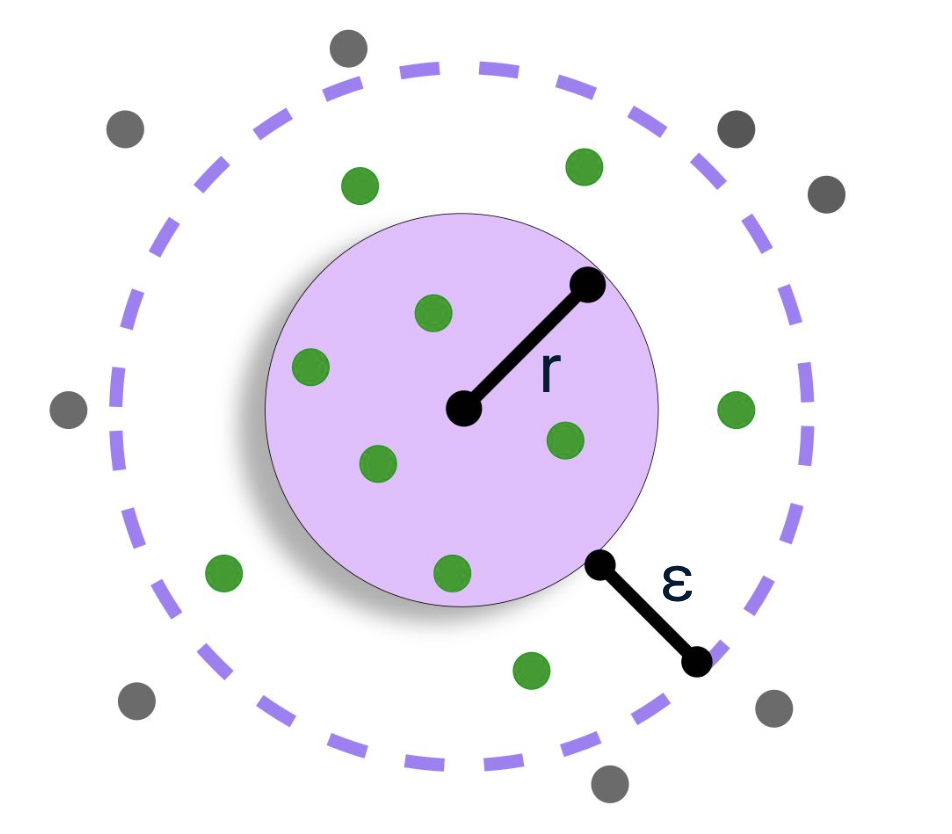}
    \caption{The solid circle shows the exact query ball of radius \( r \), while the dashed circle includes an \( \varepsilon \)-slack. All points inside the solid circle must be returned; those in the dashed region are optional under the \( (1 + \varepsilon) \)-approximation.
}
\end{figure}

When doing an exact range search, it suffices to first identify the approximate output set and then iterate through the points to see which lie strictly within the desired range.
We will use notation $k_\e$ in the output-sensitive lower bounds to indicate the number of points in the slightly larger range.  That is
\[
    k_\e := \left|P\cap \bigcap_{i=1}^m\ball_i\left(q, (1 + \e) r_i)\right) \right|
\]

For convenience, we assume $\e\le \frac{1}{2}$ throughout.
This is not strictly necessary and if one uses a larger value, one can replace all terms of the form $\left(\frac{1}{\e}\right)^{O(\delta)}$ with terms of the form $2^{O(\delta)}$.
In all cases, we have a tradeoff between the increased query time required for smaller $\e$ and the increased number of points filtered for larger $\e$.

\subsection{Doubling Dimension and Packing}
\label{sec:doubling_dimension}

Doubling dimension offers a proxy for volume in finite metric spaces.
The \emph{diameter} of a set is the maximum pairwise distance.
The \emph{doubling constant} $\rho$ is the minimum integer $\rho$ such that any ball of diameter $2r$ can be covered with $\rho$ sets of diameter at most $r$.
We define the \emph{doubling dimension} on a metric $(P, \dist)$ as $\dim(P,\dist) :=\log_2(\rho)$. 
Note that sometimes, the covering sets are required to be metric balls.
That definition is equivalent up to constant factors, but this definition allows for the following useful lemmas.

\begin{lemma}[Monotonicity of Dimension]
    If $Q\subseteq P$ and $P$ has metric $\dist$ of doubling dimension $\delta$, then $Q$ has doubling dimension at most $\delta$.
\end{lemma}

\begin{lemma}[The Dimension of Products]\label{lem:product_dimension}
Let $P$ be a finite set, with a metric $\dist$ that is the product of metrics $\dist_1$ and $\dist_2$.
Then $\dim(P,\dist) \leq \dim(P,\dist_1) + \dim(P,\dist_2)$.
\end{lemma}

A set $P$ is $r$-\emph{packed} if for all distinct $a, b \in P$, we have $\dist(a, b) \geq r$.
\begin{lemma}[Standard Packing~\cite{krauthgamer2004navigating}]\label{lem:packing}
    Let $(P, \dist)$ be a metric space and $\delta = \dim{P}$. If the set $Z$ is $r$-packed and can be covered with a ball of radius $R$ then 
    \[|Z| \leq \left(\frac{4R}{r}\right)^\delta.\]
\end{lemma}

\subsection{Greedy Permutation and Greedy Trees}

Let $P = (p_0, p_1, \cdots p_{n-1})$ be a permutation of $n$ points.
Let $\dist$ be a metric on $P$.
The $i$-th prefix is $P_i = \{p_0, p_1, \cdots p_{i-1}\}$, the set containing the first $i$ points of $P$.
Then $P$ is a \emph{greedy permutation} if for all $i>1$,
\[
    \dist (p_i, P_i) = \max_{q\in P} \dist(q, P_i).
\]
The point $p_0$ may be chosen arbitrarily.
A greedy permutation can be computed in $O(n \log \spread)$ time for low-dimensional data, where $\spread$ is the ratio of farthest to smallest pairwise distances in $P$~\cite{har2005fast}. 
There are also $O(n\log n)$ time approximation algorithms~\cite{chubet2024simple, har2005fast}.

A \emph{greedy tree} is a binary tree with the points of $P$ in the leaves.
Each node $v$ represents the cluster of points in the leaves of its subtree, denoted $\pts(v)$ and stores a center and a radius.
The greedy tree is easily constructed at the same time as the greedy permutation.
When the point $p_i$ is added to the permutation, the node of its (approximate) nearest predecessor $q$ is given two children, one centered at $p_i$ and the other centered at $q$.
Thus, every non-leaf node has a child centered at the same point, albeit with a smaller radius.

Greedy trees allow for a variety of proximity search algorithms.
The basic paradigm is that a node in the search is considered \emph{viable} if it may contain the some of the search output based on the center and the radius.
Viabe nodes are \emph{split}, replacing them with their two childrewn.
Non-viable nodes are pruned from the search.
Well-known packing and covering properties of greedy permutations allow one to bound the number of viable nodes in a search, leading to efficient running times.

A useful property of greedy trees is that they can be constructed recursively.
In particular, merging two greedy trees is more efficient than constructing a greedy tree from scratch on the union of their point sets.
\begin{lemma}(Chubet et al.~\cite{chubet2024simple})\label{lem:construction} For $n$ points in a metric of doubling dimension $\delta$, we have the following bounds.
\begin{itemize}
    \item Greedy tree construction takes $2^{O(\delta)}n\log n$ time.
    \item Merging two greedy trees takes $2^{O(\delta)} n$ time.
    \item A greedy tree requires $O(n)$ space.
\end{itemize}
\end{lemma}

\subsection{Range Search in Greedy Trees}

We analyze our range search algorithms using the proximity search framework for ball trees established by Chubet et al.~\cite{chubet2023proximity}.
A range search maintains a heap $H$ containing \emph{viable nodes} that can intersect the query range.
A range search maintains the invariant that any point in the query range is either covered by some node in $H$ or has been added to the output.
Additionally, each point appears in at most one node in $H$ or the output at any given time.
The width $w$ of a search is the maximum number of nodes in $H$ during any given iteration of the algorithm.
The height $h$ of a search is the maximum number of times the node containing any point $p$ splits.

\section{Greedy Range Tree on Product Metric}\label{sec:cascading}

A greedy range tree on a set $P$ and metrics $\dist_1, \dots, \dist_m$ is defined recursively, as follows:
\begin{itemize}
    \item The \emph{primary tree} is a greedy tree on all of $P$ using metric $\dist_1$.
    \item Each node $v$ of the primary tree points to an \emph{auxiliary} greedy range tree on $\pts(v)$ with metrics $\dist_2, \dots, \dist_m$.
\end{itemize}

\begin{figure}[h]
    \center
    \begin{subfigure}{0.45\textwidth}
    \includegraphics[width=\textwidth]{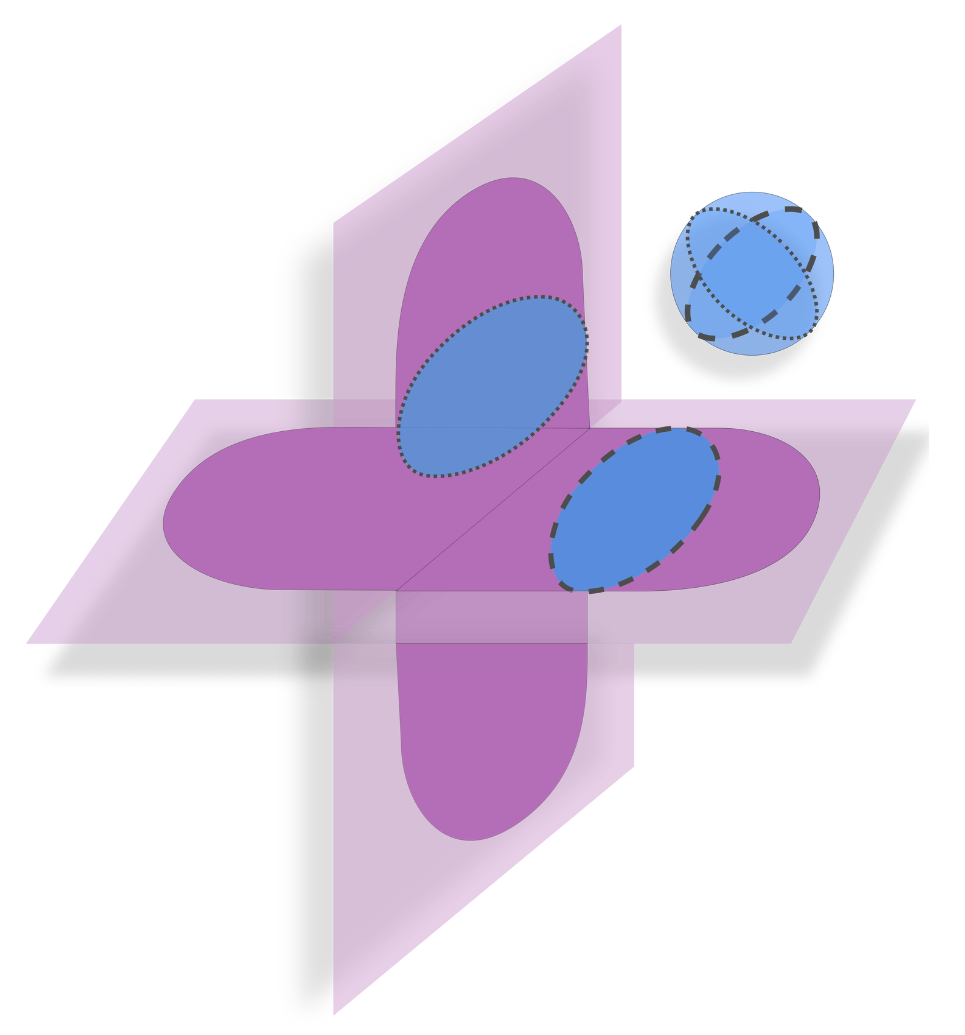}
    \caption{
    Each node of a greedy range tree corresponds to an intersection of balls from different metrics. 
    }
    \label{fig:grt}
    \end{subfigure}
    \hfill
    \begin{subfigure}{0.45\textwidth}
    \includegraphics[width=\textwidth]{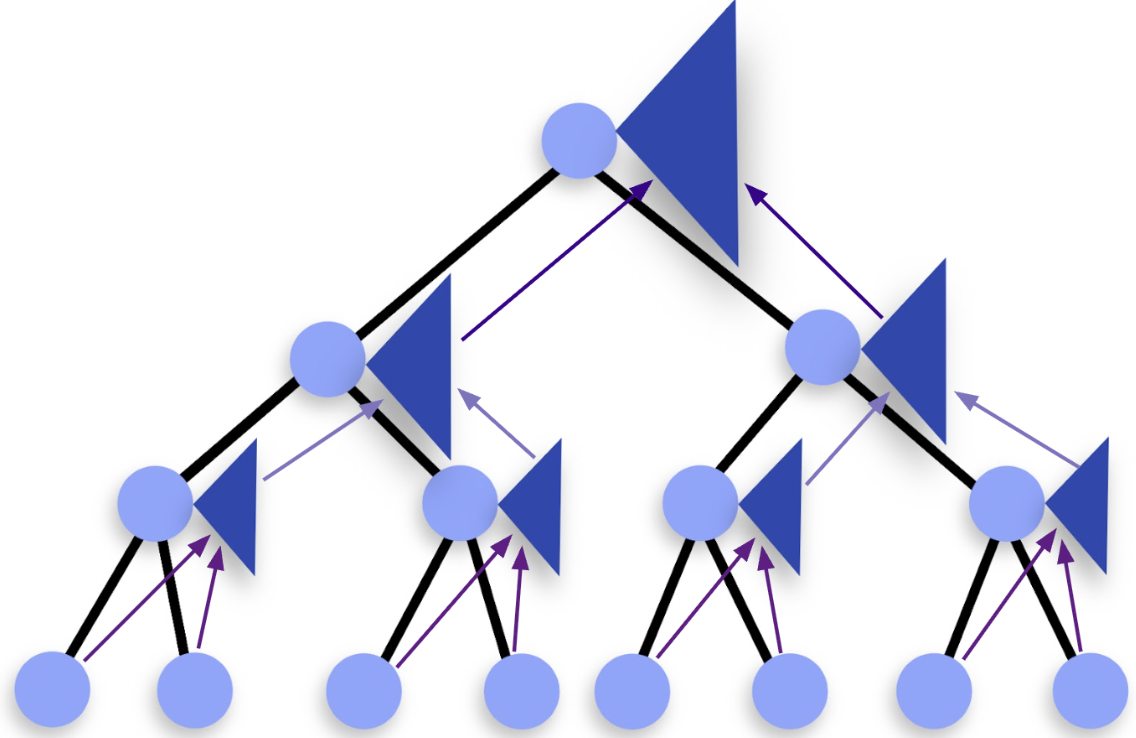}
    \caption{The auxiliary greedy range trees are built bottom-up. Each node merges the auxiliary trees of its children.}
    \label{fig:construct}
    \end{subfigure}
\end{figure}
\subsection{Construction}

The construction of a greedy range tree is also recursive (see Figure~\ref{fig:construct}).
First, we construct a greedy tree for $\dist_1$.
We observe that for a node $u$ with children $u_L$ and $u_R$, we have $\pts(u) = \pts(u_L) \cup \pts(u_R)$. 
So, rather than constructing the greedy tree on $\pts(u)$ for each node $u$, we merge the trees of its children.
The bottom-up merging process is repeated recursively.

The height of a greedy tree is $2^{\O(\delta)}\log(\spread)$~\cite{chubet2023proximity}.
Each point is contained in at most one node per depth from the root per level, so each point participates in $O(\log\spread)$ merges.
Thus, each level adds a multiplicative factor of $\log \spread$ to the construction time and space.
Note that for simplicity our analysis uses $\spread = \max_{1 \le i \le m} \spread_i$, where $\spread_i$ is the spread of $P$ under metric $\dist_i$. 
Similarly, we use $\delta = \max_{1 \le i \le m}\delta_i$, where $\delta_i$ is the doubling dimension of $P$ under metric $\dist_i$.
Then by Lemma~\ref{lem:construction}, we get the following theorem.

\begin{theorem}
Let $\dist_1, \dots, \dist_m$ be $m$ metrics on an $n$-point set $P$.
A greedy range tree can be constructed in 
$2^{O(\delta)} n \log^{m-1} \spread$
time and has space complexity $O(n \log^{m-1} \spread)$. 
\end{theorem}

\subsection{Queries}

Efficient queries in the greedy range tree exploit the fact that a standard approximate range search in a greedy tree produces a set of $(1/\e)^{O(\delta)}$ viable nodes~\cite{chubet2023proximity}.
This is essentially a cover of the desired range.
For the greedy range tree, we start with such a search using the first metric.
Instead of iterating over the points in the resulting nodes, we recursively search the auxiliary trees associated with these nodes. 

The first search of the query returns a collection of nodes that are contained in $\ball_{1}(q, (1+\e)r_1)$.
This takes $(\frac{1}{\e})^{O(\delta_1)}\log \spread$ time and returns a collection of $(\frac{1}{\e})^{O(\delta_1)}$ nodes.
Then, we query each node returned with $(q, r_2)$ in $\dist_2$.
This takes $(\frac{1}{\e})^{O(\delta_2)}\log\spread$ time per node.
We continue similarly for each metric.
Note that the resulting collection size may increase by a factor of $(\frac{1}{\e})^{O(\delta_i)}$ for each layer $i$ of the greedy range tree.
The bound on the output size of each layer follows from the bounded width of a greedy tree range search~\cite{chubet2023proximity}.

\begin{figure}[h]
    \centering
    \includegraphics[width=0.85\textwidth]{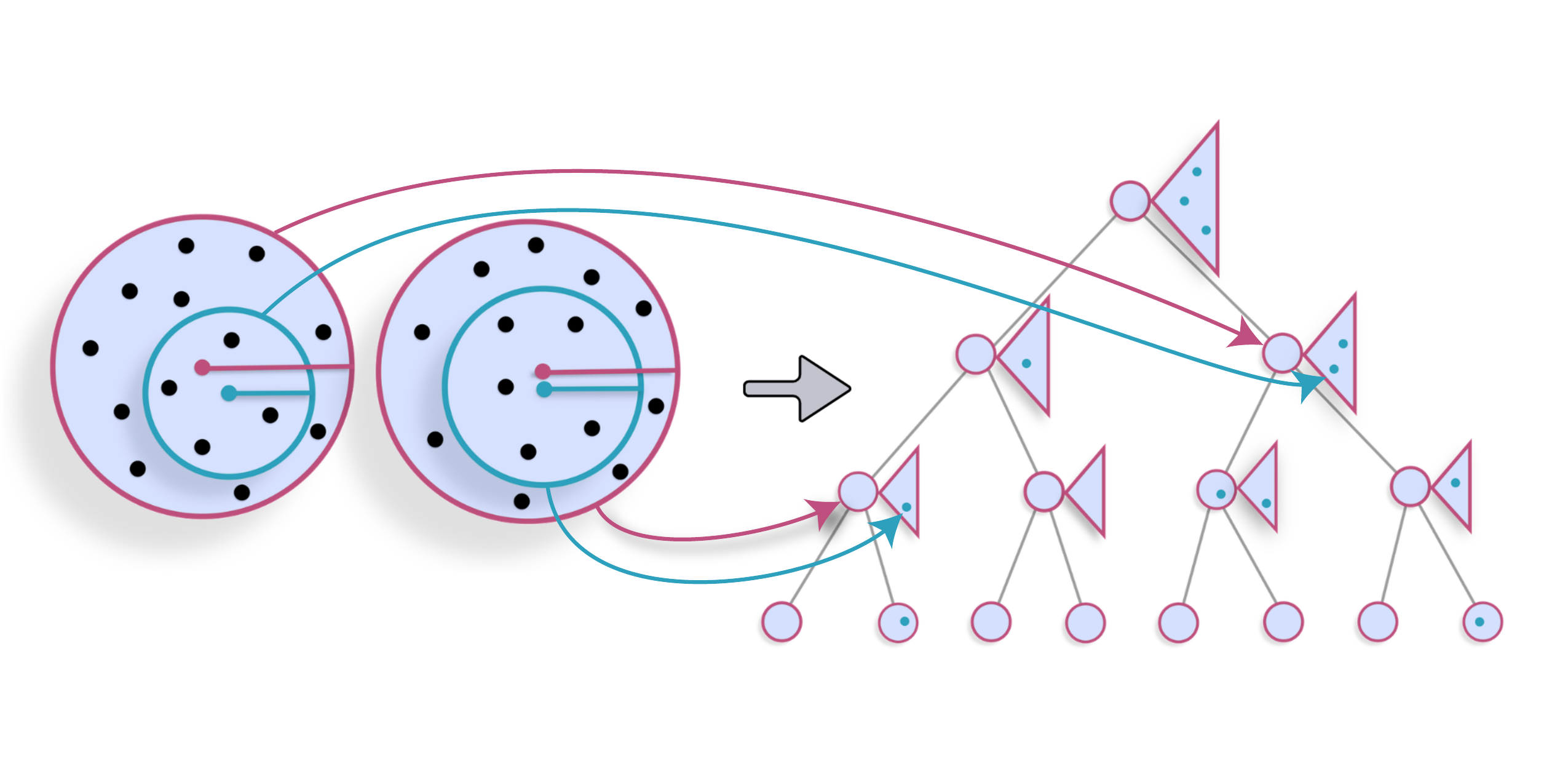}    
    \caption{Nodes (balls) in one level of the search point to trees in the next level.}
\end{figure}

\begin{theorem}
An $m$-metric query in a greedy range tree takes $(\frac{1}{\e})^{O(\delta)}m\log\spread + O(k_\e)$ time, where $\delta = \sum_{i=1}^m\delta_i$, and $k_e$ is the size of the $(1+\e)$-approximate output.
\end{theorem}

\begin{proof}
    We prove by induction on $m$ that the running time of the search without the final iteration over the viable nodes is $(\frac{1}{\e})^{O(\delta)}m\log\spread$.
    The case $m=1$ follows from the standard range search in a greedy tree~\cite{chubet2023proximity}.
    The first layer of the search takes $(\frac{1}{\e})^{O(\delta_1)}\log \spread$ time and returns a collection of $(\frac{1}{\e})^{O(\delta_1)}$ nodes.
    Each node produces an $(m-1)$-metric query.
    So, the total is
    \[
        \left(\frac{1}{\e}\right)^{O(\delta_1)}\log \spread + \left(\frac{1}{\e}\right)^{O(\delta_1)}\left(\frac{1}{\e}\right)^{O(\delta-\delta_1)}(m-1)\log\spread.
    \]
    This is at most $\left(\frac{1}{\e}\right)^{O(\delta)}m\log \spread$.
    The last step iterates over at most $k_\e$ points in the nodes in order to produce the final output, giving the desired bound.    
\end{proof}

\section{A Greedy Tree on the Product Metric}\label{sec:non-cascading}

In this section we analyze product range queries on a greedy tree using the $\ell_\infty$-product of the input metrics, rather than preprocessing each metric separately.
That is, the tree is built and traversed using the metric
\[
    \dist(x,y) = \max_{i}\dist_i(x,y).
\]
By Lemmas~\ref{lem:product_dimension} and~\ref{lem:construction}, preprocessing is $2^{O(\delta)}n\log n$, where $\delta$ is the sum of the doubling dimensions of the input metrics.
Additionally, the space is linear in the number of points.
A product range query is similar to a regular range query, but it is modified to use the different input radii.
As no change is made to the structure of the greedy tree, Lemma~\ref{lem:construction} applies.
The preprocessing time is $2^{O(\delta)}n \log n$ and the space is $O(n)$.
Below, we give the new query algorithm and its analysis.

\subsection{The Query Algorithm}

The following algorithm is similar to range search in any ball tree \cite{chubet2023proximity}, however, the pruning rules are modified to account for all query radii. 
As in Section~\ref{sec:cascading}, the output is a set of tree nodes and a final pass identifies the points in the range.

\begin{figure}
\textsc{Product Range Query}\\
\textbf{Input:} Tree $T$, query $q$, radii $r_1, \dots, r_m$, and $\e > 0$
\begin{enumerate}
    \item Initialize a max-heap $H$ to store tree nodes with the node radius as the key.
    \item Insert the root of $T$ into $H$.
    \item While $H$ is not empty:
    \begin{enumerate}
        \item Remove node $v$ with max radius $r$ from $H$.
        \item If $\forall i, \dist_i(q,v) \le (1+\e)r_i - r$, then add $v$ to $\out$.\label{algo:output_added}
        \item Else if $v$ is not a leaf:
        \begin{enumerate}
            \item Split $v$ into children $v_L$ and $v_R$ with radii $r_L$ and $r_R$ respectively.
            \item If $\forall i,\dist_i(q, v_R) \le r_i + r_R$, add $v_R$ to $H$.
            \item If $\forall i,\dist_i(q, v_L) \le r_i + r_L$, add $v_L$ to $H$.
        \end{enumerate}
    \end{enumerate}
\end{enumerate}
\end{figure}

\begin{theorem}[Correctness]
    The nodes returned by the query algorithm contain all of points in the range.
    Moreover, the nodes are contained in the $(1+\e)$-approximate range.
    That is, if $\out$ is the union of all points in the returned nodes, then
    \[
        P\cap\bigcap_{i=1}^m \ball_i(q,r_i) 
            \subseteq \out 
            \subset P\cap \bigcap_{i=1}^m \ball_i\left(q, (1 + \e) r_i)\right).
    \]
\end{theorem}

\begin{proof}
    First, we show that the algorithm maintains the invariant that if $x$ is in the query range, then either $x \in \pts(v)$ for some node $v \in \out$, or some node $v \in H$.
    We prove the invariant holds by induction.
    The root of the tree contains all points, so when $H$ is initialized the invariant holds.
    
    Assume the invariant holds at some the beginning of some iteration of the main loop.
    Consider an arbitrary point $x$ in the query range $\bigcap_{i=1}^m\ball_{d_i}(q, r_i)$.
    If $x$ has not already been added to $\out$, then let $v$ be the node in $H$ such that $x \in \pts(v)$.
    If $v$ is removed from the heap and added to $\out$, then the invariant holds for $x$.
    Otherwise, $v$ is not added to $\out$ at this step and it is split into two smaller nodes.
    Note that $x \in v_R$ or $x \in v_L$, and denote the node containing $x$ by $v'$ with radius $r'$.
    Then for all $i$, $d_i(v', q) \le d_i(v', x) + d_i(x,q) \le r' + r_i$.
    So, the algorithm adds $v'$ to $H$.
    That is, the node containing $x$ is still viable.
    The same logic applies to all points in the range, so the invariant holds after this iteration.

    Now, we show that any $x \in \out$ is approximately contained in the query range.
    Let $v \in \out$ be the node containing $x$, and let $r$ be the radius of $v$.
    If $v$ was added to $\out$ in step (\ref{algo:output_added}) then $d_i(q,v) \le (1+\e)r_i - r$ for all $i$.
    So, $d_i(x,q) \le d_i(q,v) + d_i(x,v) \le r + (1+\e)r_i - r = (1+\e)r_i$ for all $i$.
    It follows that $x\in \bigcap_{i=1}^m \ball_i\left(q, (1 + \e) r_i)\right)$.
\end{proof}

\subsection{The Query Time}

\begin{figure}[h]
    \center
    \includegraphics[width=0.8\textwidth]{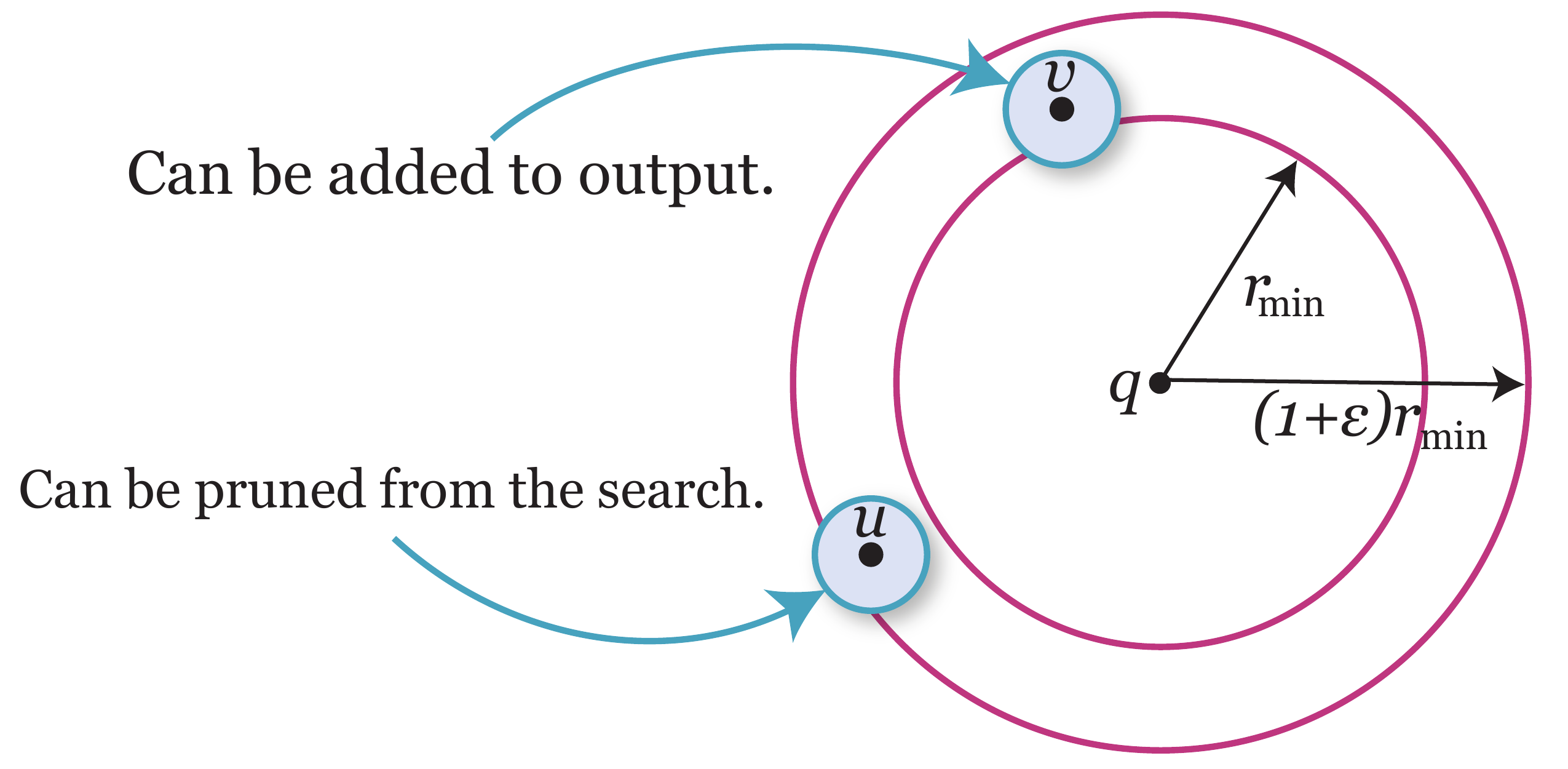}
    \caption{Every node with radius at most $\frac{\e}{2}r_{\min}$ is either immediately added to the output (in step 2) or will be pruned from the search (in step 3).\label{fig:in_or_out}}
\end{figure}


\begin{theorem}
Let $d_1, d_2, \ldots, d_m$ be metrics with doubling dimensions
$\delta_1, \delta_2, \ldots, \delta_m$ respectively.
Then a product range query $(q, r_1, r_2, \ldots, r_m, \e)$ takes
\[
    \left(\frac{A}{\e}\right)^{O(\delta)}\log\spread + O(k_\e)
\]
time, where $\e<\frac{1}{2}$, $\delta := \sum_{i=1}^m \delta_i$, $k_\e$ is the size of the $(1+\e)$-approximate output, and $A = \frac{\max_{i} r_i}{\min_{j} r_j}$.

\end{theorem}

\begin{proof}
    Using the framework of~\cite{chubet2023proximity}, a range query in the greedy tree takes $wh\log w + O(k_\e)$ time where $w$, the \emph{width}, is the maximum cardinality of $H$ throughout the algorithm and $h$, the \emph{height}, is the maximum number of times, over the points $x \in P$, that the node containing $x$ is split.
    The height $h$ is bounded from above by the height of the tree, which is $2^{O(\delta)}\log\spread$ by Lemma~\ref{lem:packing}.
    So it is sufficient to bound the width.

    At any time in the algorithm, let $r$ be the maximum radius of a node in $H$.
    Let $r_{\max} := \max_{1\le i \le m}r_i$, and $r_{\min} := \min_{1\le i\le m}r_i$.
    
    For any node $v \in H$, $v \in \ball(q, r + r_{\max})$.
    Additionally, the algorithm guarantees that $r > \frac{\e}{2} r_{\min}$.
    Anything smaller than that and the node is either pruned or sent to the output (See Figure~\ref{fig:in_or_out}).
    As shown in~\cite{chubet2023proximity}, the nodes of $H$ will always be $cr$-packed for some fixed constant $c\le 1$.
    So by the Standard Packing (Lemma~\ref{lem:packing}), we have
    \[
        w \le \left(4 \frac{r+ r_{\max}}{cr}\right)^{\delta} \le \left(\frac{4}{c} + \frac{r_{\max}}{c\e r_{\min}}\right)^\delta = \left(\frac{A}{\e}\right)^{O(\delta)}.
    \]
    So, with a width $\left(\frac{A}{\e}\right)^{O(\delta)}$ and a height $2^{O(d)}\log \spread$, we get the desired bound
    \[
        wh\log w + O(k_\e) = \left(\frac{A}{\e}\right)^{O(\delta)}\log\spread + O(k_\e).
    \]
\end{proof}

\section{Conclusion}
The overall goal of this paper was to present two new data structures for approximate range search in product metric spaces, using greedy permutations to unify multiple metrics into a single searchable structure. Our results show that it's possible to keep queries fast and storage efficient, even when working across several metrics. While adding or removing a single point currently requires rebuilding the structure, batching updates, such as inserting or deleting many points at once, can be handled more efficiently. In future work, we plan to implement these structures, test their performance, and explore practical strategies for supporting dynamic updates. Overall, these tools offer a simple and scalable approach to multi-metric search that is both theoretically sound and useful in practice.

\subsection{Applications}

Many real-world datasets involve multiple notions of similarity that a single metric doesn't fully capture. For instance, recommendation systems may compare users based on both temporal behavior and content preferences \cite{ieee7498318}, while molecular data in biology might depend on both spatial structure and chemical properties. Multi-metric range search offers a way to filter such data using multiple criteria simultaneously, similar to how sensor networks can query events across multiple attributes like temperature and light levels \cite{LiKimGovindanHong2003}.

\vspace{3mm}

\noindent The data structures proposed in this paper are well-suited for these settings. They support efficient approximate range queries across multiple metrics without requiring separate structures for each one. The Greedy Range Tree is particularly useful when metrics vary in cost or behavior, allowing for progressive filtering. The Single Greedy Tree, while it is simpler to construct, performs well when the product metric is well-structured and balanced.

\subsection{Acknowledgements}\label{sec:acknowledgements}

The authors thank the U.S. National Science Foundation for supporting this project as part of the Research Experience for Undergraduates program under grant CCF-2349179.
We also thank the anonymous reviewers for their helpful suggestions in improving the submitted draft.
\bibliographystyle{abbrv}




\end{document}